\pdfoutput=1 

\documentclass[final, nomarks]{dmtcs-episciences}

\usepackage{amsthm}
\usepackage{mathtools}
\usepackage{microtype}

\def\ttt{\texttt}

\def\Zee{{\mathbb{Z}}}

\DeclareMathOperator{\alp}{alph}
\DeclareMathOperator{\append}{append}
\DeclareMathOperator{\depth}{depth}
\DeclareMathOperator{\level}{level}
\DeclareMathOperator{\rank}{rank}
\DeclareMathOperator{\short}{short}

\newtheorem{theorem}{Theorem}
\newtheorem{lemma}{Lemma}
\newtheorem{proposition}{Proposition}

\theoremstyle{definition} 
\newtheorem{example}{Example}

\title{A Characterization of Morphic Words with Polynomial Growth}

\author{Tim Smith\affiliationmark{1}}

\affiliation{School of Computer Science, University of Waterloo, Canada}

\keywords{morphic word, polynomial growth, zigzag word, multilinear word}

\received{2019-3-29}
\revised{2019-11-1}
\accepted{2019-12-1}

\publicationdetails{22}{2019}{1}{3}{5324}

\begin{document}

\maketitle

\begin{abstract}
A morphic word is obtained by iterating a morphism to generate an infinite word, and then applying a coding.  We characterize morphic words with polynomial growth in terms of a new type of infinite word called a \textit{zigzag word}.  A zigzag word is represented by an initial string, followed by a finite list of terms, each of which repeats for each $n \geq 1$ in one of three ways: it grows forward [$t(1)\ t(2)\ \dotsm\ t(n)]$, backward [$t(n)\ \dotsm\ t(2)\ t(1)$], or just occurs once [$t$].  Each term can recursively contain subterms with their own forward and backward repetitions.  We show that an infinite word is morphic with growth $\Theta(n^k)$ iff it is a zigzag word of depth $k$.  As corollaries, we obtain that the morphic words with growth $O(n)$ are exactly the ultimately periodic words, and the morphic words with growth $O(n^2)$ are exactly the multilinear words.
\end{abstract}

\section{Introduction}

Morphic words \cite{as} are a well-studied class of infinite words obtained by iterating a morphism $h$ on a letter $c$, and then applying a coding (letter-to-letter mapping) $\tau$.  The morphism $h$ is required to be prolongable on $c$, meaning that $h(c) = cx$ for some string $x$ such that $h^i(x)$ is nonempty for all $i$.  Such a word has the form
\begin{align*}
\tau(h^\omega(c)) = \tau(c\ x\ h(x)\ h^2(x)\ h^3(x)\ \dotsm)
\end{align*}
We call the triple $(h,c,\tau)$ a ``representation'' of the morphic word $\alpha = \tau(h^\omega(c))$.  In this paper we characterize morphic words with polynomial growth, i.e. those having a representation $(h,c,\tau)$ whose growth function $f(n) = |h^n(c)| = |\tau(h^n(c))|$ is bounded by a polynomial in $n$.

Our characterization involves a new type of infinite word which we call a \textit{zigzag word}.  A zigzag word is represented by an initial string (which may be empty), followed by a finite list of terms.  Each term repeats for each $n \geq 1$ according to one of three functions, denoted $F$, $B$, and $S$.  The function $F$ causes forward growth, $B$ causes backward growth, and $S$ causes stasis (no growth).  Each $F$ or $B$ term can recursively contain subterms with their own instances of the three functions, while each $S$ term contains only a string.  For example, the list
\begin{align*}
l = [(B,[(S,\ttt{a}),(F,[(S,\ttt{b})])])]
\end{align*}
represents the zigzag word
\begin{align*}
\alpha &= \prod_{n \geq 1} \prod_1^{i=n} ( \ttt{a} \prod_{j=1}^n \ttt{b} ) \\
			 &= \ttt{ab abb ab abbb abb ab abbbb abbb abb ab }\dotsm
\end{align*}
Here, $\prod$ denotes concatenation, and we use $\prod_{i=1}^n f(i)$ to mean $f(1) f(2) \dotsm f(n)$, and $\prod_1^{i=n} f(i)$ to mean $f(n) f(n{-}1) \dotsm f(1)$.  Thus, the outermost $\prod$ corresponds to the increasing bound $n$ which is always present, the middle $\prod$ corresponds to the $B$ and indicates backwards growth from $n$ down to 1, and the innermost $\prod$ corresponds to the $F$ and indicates forwards growth from 1 up to $n$.  The net effect is that the number of \ttt{a}s in each block remains stationary at 1, while the number of \ttt{b}s ``zigzags'' with the pattern 1, 2, 1, 3, 2, 1, 4, 3, 2, 1, $\dotsc$  A formal definition of zigzag words appears in Section \ref{sec:zigzag}, along with the notion of the depth of a zigzag representation as the level of nesting of its terms.

As our main result, we show that an infinite word has a morphic representation with growth $\Theta(n^k)$ iff it has a zigzag representation of depth $k$.  Our proof makes use of the notion of ``rank'' due to Ehrenfeucht and Rozenberg \cite{er79}.  Notice that whereas a morphic representation makes use of a coding, a zigzag representation does not: the letters in the zigzag representation are just those that appear in the word.

We use our result to make connections with two other classes of infinite words: ultimately periodic words and multilinear words.  Ultimately periodic words have the form $qrrr\dotsm$ for strings $q,r$, while multilinear words \cite{smith3,ehk} have the form
\begin{align*}
q\prod\limits_{n\geq 0}r_1^{a_1n+b_1}r_2^{a_2n+b_2} \dotsm r_m^{a_mn+b_m}
\end{align*}
where $\prod$ denotes concatenation, $q,r_i$ are strings, and $a_i,b_i$ are nonnegative integers.  For example, $\prod\limits_{n\geq 0}\texttt{a}^{n+1}\texttt{b}$ = $\texttt{abaabaaab}\dotsm$ is a multilinear word.  We show that the morphic words with growth $O(n)$ are exactly the ultimately periodic words, and that the morphic words with growth $O(n^2)$ are exactly the multilinear words.  Thus, ultimately periodic words and multilinear words can be seen as the first two levels of the hierarchy of morphic words with polynomial growth, or equivalently, of the hierarchy of zigzag words.

\subsection{Related work}

A previous characterization of morphisms with polynomial growth is due to Ehrenfeucht and Rozenberg \cite{er79}, who introduce the notion of the ``rank'' of letters under a morphism.  They show that a string $x$ has rank $k$ under a morphism $h$ iff $k$ is the minimal degree of a polynomial $p$ such that for every $n$, $p(n) \geq |h^n(x)|$.  We make use of this result in proving the equivalence between morphic words with polynomial growth and zigzag words.

Other topics studied in connection with morphisms with polynomial growth include questions of sequence equivalence \cite{karhumaki1977,honkala2003a}, $\omega$-equivalence \cite{honkala2003b}, the relationship between these types of equivalence \cite{honkala2002}, language equivalence \cite{honkala2004}, length sets \cite{ekr1978}, codes \cite{hw1985}, and boundedness \cite{hl1987}.

The class of multilinear words appears in \cite{smith3} as the infinite words determined by one-way stack automata, and also in \cite{ehk} (as the reducts of the ``prime'' stream $\mathrm{\Pi}$).  In \cite{smith2016b}, prediction of periodic words and multilinear words is studied in an automata-theoretic setting.

\subsection{Outline of the paper}

The rest of the paper is organized as follows.  Section \ref{sec:preliminaries} gives preliminary definitions.  Section \ref{sec:zigzag} defines zigzag words and gives examples.  Section \ref{sec:results} proves our main result, the correspondence between zigzag words and morphic words with polynomial growth.  Section \ref{sec:applications} applies this correspondence in connection with ultimately periodic and multilinear words.  Section \ref{sec:conclusion} gives our conclusions.

\section{Preliminaries}\label{sec:preliminaries}

We denote the positive integers by $\Zee^+$.  When $X$ is a set, we denote the cardinality of $X$ by $|X|$.  We use square brackets to denote a list: e.g. $[x_1, \dotsc, x_m]$ is the list containing the elements $x_1, \dotsc, x_m$ in that order.  The number of elements in a list $v$ is denoted by $|v|$.  For lists $v_1, \dotsc, v_n$, we let $\append(v_1, \dotsc, v_n)$ denote the list of length $|v_1| + \dotsb + |v_n|$ consisting of the elements of $v_1$, followed by the elements of $v_2$, $\dotsc$, followed by the elements of $v_n$.

An \textbf{alphabet} $A$ is a finite set of letters.  A \textbf{word} is a concatenation of letters from $A$.  We denote the set of finite words by $A^*$ and the set of infinite words by $A^\omega$.  We call finite words \textbf{strings}.  The length of a string $x$ is denoted by $|x|$.  We denote the empty string by $\lambda$.  We write $A^+$ to mean $A^* - \{\lambda\}$.  A \textbf{language} is a subset of $A^*$.  For a finite or infinite word $S$, a \textbf{prefix} of $S$ is a string $x$ such that $S = x S'$ for some word $S'$.  The $i$th letter of $S$ is denoted by $S[i]$; indexing starts at $1$.  We denote the set of letters occurring in $S$ by $\alp(S)$.

\paragraph{Periodic and multilinear words}

For a nonempty string $x$, $x^\omega$ denotes the infinite word $xxx\dotsm$.  Such a word is called \textbf{purely periodic}.  An infinite word of the form $xy^\omega$, where $x$ and $y$ are strings and $y \neq \lambda$, is called \textbf{ultimately periodic}.  An infinite word is \textbf{multilinear} if it has the form
\begin{align*}
q\prod\limits_{n\geq 0}r_1^{a_1n+b_1}r_2^{a_2n+b_2} \dotsm r_m^{a_mn+b_m}
\end{align*}
where $\prod$ denotes concatenation, $q$ is a string, $m$ is a positive integer, and for each $1 \leq i \leq m$, $r_i$ is a nonempty string and $a_i$ and $b_i$ are nonnegative integers such that $a_i + b_i > 0$.  For example, $\prod\limits_{n\geq 0}\texttt{a}^{n+1}\texttt{b}$ = $\texttt{abaabaaab}\dotsm$ is a multilinear word.  Clearly the multilinear words properly include the ultimately periodic words.  Any multilinear word that is not ultimately periodic we call \textbf{properly multilinear}.

\paragraph{Morphic words}

A \textbf{morphism} on an alphabet $A$ is a map $h$ from $A^*$ to $A^*$ such that for all $x,y \in A^*$, $h(xy) = h(x)h(y)$.  Notice that $h(\lambda) = \lambda$.  The morphism $h$ is a \textbf{coding} if for all $c \in A$, $|h(c)| = 1$.  For $x \in A^*$, we let $L(x,h)$ denote the set of strings $\{h^i(x) \mid i \geq 0\}$.  The letter $c$ is \textbf{recursive} (for $h$) if for some $i \geq 1$, $h^i(c)$ contains $c$.  A string $x \in A^*$ is \textbf{mortal} (for $h$) if there is an $m \geq 0$ such that $h^m(x) = \lambda$.  The morphism $h$ is \textbf{prolongable} on a letter $c$ if $h(c) = cx$ for some $x \in A^*$, and $x$ is not mortal.  If $h$ is prolongable on $c$, $h^\omega(c)$ denotes the infinite word $c\ x\ h(x)\ h^2(x)\ \dotsm$.  We call such an infinite word \textbf{pure morphic}.  An infinite word $\alpha$ is \textbf{morphic} if there is a morphism $h$, coding $\tau$, and letter $c$ such that $h$ is prolongable on $c$ and $\alpha = \tau(h^\omega(c))$.  For example, let
\[
\begin{matrix*}[l]
h(\texttt{c}) = \texttt{cbaa} & \ \tau(\texttt{c}) = \texttt{a} \\
h(\texttt{a}) = \texttt{aa} & \ \tau(\texttt{a}) = \texttt{a} \\
h(\texttt{b}) = \texttt{b} & \ \tau(\texttt{b}) = \texttt{b}
\end{matrix*}
\]
Then $\tau(h^\omega(\texttt{c})) = \texttt{a}^1\texttt{ba}^2\texttt{ba}^4\texttt{ba}^8\texttt{ba}^{16}\texttt{b}\dotsm$ is a morphic word.  See Allouche and Shallit \cite{as} for more on morphic words.

A morphism $h$ has \textbf{growth} $f(n)$ on a string $x$ if $|h^n(x)| = f(n)$ for all $n \geq 0$.  We say that $h$ is \textbf{polynomially bounded} on $x$ if $|h^n(x)|$ is in $O(n^k)$ for some $k \geq 0$.  The following proposition says that if $h$ is polynomially bounded on $x$ (and $x$ is not mortal), its growth on $x$ must be in $\Theta(n^k)$ for some $k \geq 0$, and so cannot be an ``in-between'' function like $n \log n$.

\begin{proposition}\label{prop1}For every morphism $h$ and string $x$, if $h$ is polynomially bounded on $x$ and $x$ is not mortal under $h$, then $|h^n(x)|$ is in $\Theta(n^k)$ for some $k \geq 0$.\end{proposition}
\begin{proof}
Take any morphism $h$ and string $x$ such that $h$ is polynomially bounded on $x$ and $x$ is not mortal under $h$.  Take the lowest $k$ such that $|h^n(x)|$ is in $O(n^k)$.  By \cite[Theorem 3]{er79}, $x$ has rank $k$ under $h$.  Then by \cite[Corollary 1]{er79}, we have that $|h^n(x)|$ is in $\Theta(n^k)$.
\end{proof}

We say that a morphic word $\alpha$ has growth $f(n)$ if for some morphism $h$, coding $\tau$, and letter $c$, $h$ is prolongable on $c$, $\alpha = \tau(h^\omega(c))$, and $h$ has growth $f(n)$ on $c$.  We say that $\alpha$ has polynomial growth if it has growth $\Theta(n^k)$ for some $k$.

Note that a morphic word with polynomial growth may have alternative representations in which growth is exponential.  For example, $\ttt{a}^\omega$ has polynomial growth (take $c = \ttt{s}$, $h(\ttt{s}) = \ttt{sa}$, $h(\ttt{a}) = \ttt{a}$, $\tau(\ttt{s}) = \ttt{a}$, $\tau(\ttt{a}) = \ttt{a}$), notwithstanding the existence of exponential representations (e.g. $c = \ttt{a}$, $h(\ttt{a}) = \ttt{aa}$, $\tau(\ttt{a}) = \ttt{a}$).  In conjunction with Proposition 1 above, \cite[Theorem 25]{dk2009} shows that for every aperiodic pure morphic word $\alpha$, either (1) every representation $(h,c)$ of $\alpha$ has exponential growth, or (2) for some $k \geq 1$, every representation $(h,c)$ of $\alpha$ has growth $\Theta(n^k)$.  We do not know whether the same holds for (not necessarily pure) morphic words and representations $(h,c,\tau)$.

\section{Zigzag words}\label{sec:zigzag}

We now introduce zigzag words, define a notion of depth for these words, and give some examples.  Below, $\prod$ denotes concatenation, and we use $\prod_{i=1}^n f(i)$ to mean $f(1) f(2) \dotsm f(n)$, and $\prod_1^{i=n} f(i)$ to mean $f(n) f(n{-}1) \dotsm f(1)$.  In this section and all following ones, let $A$ be an alphabet.  Let $F$, $B$, and $S$ be functions, to be defined below.

Let $L$ be the set of all nonempty lists over $(\{F,B\} \times L) \cup (\{S\} \times A^+)$.  That is, $L$ consists of all lists of the form
\[
	[(f_1,x_1), \dotsc, (f_m,x_m)]
\]
with $m \geq 1$, such that for each $1 \leq i \leq m$,
\begin{align*}
	&f_i \text{ is in } \{F,B\} \text{ and } x_i \text{ is in } L, \text{ or} \\
	&f_i = S \text{ and } x_i \text{ is in } A^+.
\end{align*}
Define $R : L \times \Zee^+ \rightarrow A^+$ as follows:
\begin{align*}
	R([(f_1,x_1),\dotsc,(f_m,x_m)], i) = \prod_{j=1}^m f_j(x_j,i)
\end{align*}
Define $F : L \times \Zee^+ \rightarrow A^+$ as follows:
\begin{align*}
	F(l,n) = \prod_{i=1}^n 	R(l,i)
\end{align*}
Define $B : L \times \Zee^+ \rightarrow A^+$ as follows:
\begin{align*}
	B(l,n) = \prod_1^{i=n} 	R(l,i)
\end{align*}
Define $S : A^+ \times \Zee^+ \rightarrow A^+$ as follows:
\begin{align*}
	S(r, n) = r
\end{align*}
A \textbf{zigzag word} is an infinite word $\alpha$ such that for some $q \in A^*$ and $l \in L$,
\begin{align*}
		\alpha = q \prod_{i \geq 1} R(l,i)
\end{align*}							

\subsection{Depth of a zigzag word}

For a list $l = [(f_1,x_1), \dotsc, (f_m,x_m)] \in L$, we define
\begin{align*}
	\depth(l) = \max \{\depth(i) \mid 1 \leq i \leq m\}
\end{align*}
where in the context of $l$, for each $1 \leq i \leq m$, we define
\begin{align*}
	\depth(i) =
		\begin{cases}
			1 &\text{ if } f_i = S \\
			\depth(x_i)+1 &\text{ if } f_i = F \text{ or } B
		\end{cases}
\end{align*}
A zigzag word $\alpha$ has depth $k$ if for some $q \in A^*$ and $l \in L$, 
\begin{align*}
		\alpha = q \prod_{i \geq 1} R(l,i)
\end{align*}
and $\depth(l) = k$.

\subsection{Examples of zigzag words}

Below we give examples of zigzag words $\alpha$ of the form $q \prod_{i \geq 1} R(l,i)$ for various values of $q$ and $l$.  For each word, we also give a shorthand notation of the form $q:r$, where $r$ is a string obtained as follows.  For clarity, we enclose string literals in quotes in the following definitions.  Define $r = \short(l)$, where
\begin{align*}
 \short([(f_1,x_1), \dotsc, (f_m,x_m)]) = \prod_{1 \leq i \leq m} \short(f_i,x_i)
\end{align*}
and where
\begin{align*}
	\short(f_i,x_i) =
		\begin{cases}
			x_i &\text{ if } f_i = S \\
			\text{``('' }\short(x_i)\text{ ``)''} &\text{ if $f_i$ is in $\{F,B\}$ and $\depth(x_i) = 1$} \\
			\text{``F('' }\short(x_i)\text{ ``)''} &\text{ if $f_i = F$ and $\depth(x_i) > 1$} \\
			\text{``B('' }\short(x_i)\text{ ``)''} &\text{ if $f_i = B$ and $\depth(x_i) > 1$}
		\end{cases}
\end{align*}
We write the shorthand $q : r$ as just $r$ if $q = \lambda$.

\begin{example} (depth 1, ultimately periodic)
\begin{align*}
	q = \ttt{a}, l = [(S,\ttt{bc})]  \\
		\alpha  =  \ttt{a}(\ttt{bc})^\omega  =  \ttt{abcbcbc}\dotsm \\
 \text{shorthand:\ \ } \ttt{a}:\ttt{bc}
\end{align*}
\end{example}

\begin{example} (depth 2, multilinear)
\begin{align*}
	q = \lambda, l = [(S,\ttt{a}),(F,[(S,\ttt{b})])]  \\
		\alpha = \prod_{n \geq 1} \ttt{ab}^n = \ttt{ababbabbb}\dotsm \\
\text{shorthand:\ \ } \ttt{a}(\ttt{b})
\end{align*}
\end{example}

\begin{example} (depth 3)
\begin{align*}
	q = \lambda, l = [(F,[(S,\ttt{a}),(F,[(S,\ttt{b})])]), (B,[(S,\ttt{c}),(F,[(S,\ttt{d})])])] \\
		\alpha = \prod_{n \geq 1} (\prod_{i=1}^n \ttt{ab}^i) (\prod_1^{i=n} \ttt{cd}^i) \\
 		  = 	\ttt{abcd ababbcddcd ababbabbbcdddcddcd}\ \dotsm \\
\text{shorthand:\ \ } \text{F}(\ttt{a}(\ttt{b}))\ \text{B}(\ttt{c}(\ttt{d}))
\end{align*}
\end{example}

\section{Equivalence of morphic words with polynomial growth and zigzag words}\label{sec:results}

In this section we establish that an infinite word is morphic with growth $\Theta(n^k)$ iff it is a zigzag word of depth $k$ (Theorem \ref{theorem-main}).  We first show that every zigzag word of depth $k$ is a morphic word with growth $\Theta(n^k)$ (Theorem \ref{zigzag-morphic}), and then that every morphic word with growth $\Theta(n^k)$ is a zigzag word of depth $k$ (Theorem \ref{morphic-zigzag}).

\subsection{From zigzag words to morphic words}

\begin{lemma}\label{lemma1}For every $l \in L$, there is a string $w$, morphism $h$, and coding $\tau$ such that for all $n \geq 0$, $\tau(h^n(w)) = R(l,n+1)$, and $h$ has growth $\Theta(n^{\depth(l)-1})$ on $w$.\end{lemma}

\begin{proof}
The list $l$ has the form $[(f_1,x_1), \dotsc, (f_m,x_m)]$.
	
We proceed by induction on the depth $k$ of $l$.
	
If $k = 1$, then every $f_i = S$, so for all $n \geq 1$, $R(l,n) = x_1 \dotsm x_m$.  Then we can take $w = x_1 \dotsm x_m$, and for every letter $c$ in $w$, set $h(c) = \tau(c) = c$.  The morphism $h$ has growth $\Theta(n^0) = \Theta(1)$ on $w$ as desired.
	
If $k > 1$, suppose for induction that the claim is true for every list of depth $< k$.
	
For each $(f_j,x_j)$ with $1 \leq j \leq m$, we will describe how to construct a string $w_j$, morphism $h_j$, and coding $\tau_j$ so that $h_j$ has growth $\Theta(1)$ on $w_j$ if $f_j = S$ and growth $\Theta(n^{\depth(x_j)})$ on $w_j$ if $f_j$ is in $\{F,B\}$, and further that for all $n \geq 0$, $\tau_j(h_j^n(w_j)) = f_j(x_j,n+1)$.
	
If $f_j = S$, then $x_j$ is a string.  So set $w_j = x_j$, and for all $c$ in $x_j$, set $h_j(c) = \tau_j(c) = c$.  Then $h_j$ has growth $\Theta(1)$ on $w_j$ as desired.
	
If $f_j = F$, then since $\depth(x_j) < \depth(l)$, we can apply the induction hypothesis, obtaining $w'_j,h'_j,\tau'_j$ such that for all $n \geq 0$, $\tau'_j(h'^{\,n}_j(w'_j)) = R(x_j,n+1)$, and $h'_j$ has growth $\Theta(n^{\depth(x_j)-1})$ on $w'_j$.  For all $c$ in the alphabet of $h'_j$, set $h_j(c) = h'_j(c)$ and for all $c$ in the alphabet of $\tau'_j$, set $\tau_j(c) = \tau'_j(c)$.  Now, let $a$ be a new letter.  Set $w_j = a\ w'_j[2] \dotsm w'_j[|w'_j|]$.  Set $h_j(a) = w_j\ h_j(w'_j[1])$ and set $\tau_j(a) = \tau'_j(w'_j[1])$.  Then $h_j(w_j) = w_j\ h_j(w'_j)$, so we have for all $n \geq 0$,
\begin{align*}
	\tau_j(h_j^n(w_j)) &= \tau_j(w_j \prod_{i=1}^n h_j^i(w'_j)) \\
									&= \tau'_j(w'_j \prod_{i=1}^n h'^{\,i}_j(w'_j)) \\
									&= \prod_{i=0}^n \tau'_j(h'^{\,i}_j(w'_j)) \\
									&= \prod_{i=1}^{n+1} R(x_j,i) \\
									&= F(x_j,n+1)
\end{align*}
as desired.  From above, we have that for all $n \geq 0, |h_j^n(w_j)| = |\prod_{i=0}^n h'^{\,i}_j(w'_j)| = \sum_{i=0}^n |h'^{\,i}_j(w'_j)|$.  We know that $h'_j$ has growth $\Theta(n^{\depth(x_j)-1})$ on $w'_j$, so $|h'^{\,n}_j(w'_j)|$ is bounded above and below by polynomials of degree $\depth(x_j)-1$.  Therefore $h_j$ has growth $\Theta(n^{\depth(x_j)})$ on $w_j$, since for any polynomial $p(n)$ of degree $k$, the sum $\sum_{i=0}^n p(i)$ equals a polynomial $p'(n)$ of degree $k+1$.

If $f_j = B$, we proceed as in the case $f_j = F$, but this time set $w_j = w'_j[1] \dotsm w'_j[|w'_j|-1]\ a$, set $h_j(a) = h_j(w'_j[|w'_j|])\ w_j$, and set $\tau_j(a) = \tau'_j(w'_j[|w'_j|])$.  Then $h_j(w_j) = h_j(w'_j)\ w_j$, so for all $n \geq 0$,
\begin{align*}
	\tau_j(h_j^n(w_j)) &= \tau_j(\prod_1^{i=n} h_j^i(w'_j)\ \ w_j) \\
&= \tau'_j(\prod_1^{i=n} h'^{\,i}_j(w'_j)\ \ w'_j) \\
&= \prod_0^{i=n} \tau'_j(h'^{\,i}_j(w'_j)) \\
&= \prod_1^{i=n+1} R(x_j,i) \\
&= B(x_j,n+1)
\end{align*}
as desired.  Further, by the same reasoning as for the case $f_j = F$, $h_j$ has growth $\Theta(n^{\depth(x_j)})$ on $w_j$.

We now have $w_j,h_j,\tau_j$ for each $j$, and we want to make a unified $w,h,\tau$ for the whole list $l$.  First, we rename certain letters to avoid conflicts.  We say that a conflict occurs when there are $j_1 \neq j_2$ and a letter $c$ such that $c$ belongs to the alphabet of both $h_{j_1}$ and $h_{j_2}$.  If there is a conflict, take any such $c,j_1,j_2$.  Let $d$ be a new letter not appearing in any $w_j$, $h_j$, or $\tau_j$.  Replace all occurrences of $c$ in $w_{j_1}$, $h_{j_1}$, and the lefthand side of $\tau_{j_1}$ with $d$.  Repeat this process until no conflicts remain.

With all conflicts resolved, we set $w = w_1 \dotsm w_m$, and create a morphism $h$ and coding $\tau$ such that for every $1 \leq j \leq m$, for every $c$ in the alphabet of $h_j$, $h(c) = h_j(c)$ and $\tau(c) = \tau_j(c)$.  Now we have that for all $n \geq 0$,
	\begin{align*}
	\tau(h^n(w))	= \prod_{j=1}^m f_j(x_j,n+1) = R(l,n+1)
	\end{align*}
as desired.  Now, since $\depth(l) = k$ and $k > 1$, we have by definition that for every $j$ with $f_j$ in $\{F,B\}$, $\depth(x_j) \leq k-1$, and for at least one such $j$, $\depth(x_j) = k-1$.  Then by our construction, every $h_j$ has growth $O(n^{k-1})$ on $w_j$, and at least one $h_j$ has growth $\Theta(n^{k-1})$ on $w_j$.  Therefore $h$ has growth $\Theta(n^{k-1})$ on $w$.
\end{proof}

\begin{theorem}\label{zigzag-morphic}Every zigzag word of depth $k$ is a morphic word with growth $\Theta(n^k)$.\end{theorem}

\begin{proof}
Take any zigzag word $\alpha$ of depth $k$.  Then for some $q \in A^*$ and $l \in L$ with $\depth(l) = k$,
\begin{align*}
	\alpha = q \prod_{i \geq 1} R(l,i)
\end{align*}
By Lemma \ref{lemma1}, there are $w,h,\tau$ such that for all $n \geq 0$, $\tau(h^n(w)) = R(l,n+1)$, and $h$ has growth $\Theta(n^{k-1})$ on $w$.  Let $a,b,c_1,\dotsc,c_{|q|},d_1,\dotsc,d_{|w|}$ be new letters.  Let $s = c_1 \dotsm c_{|q|} \ d_1 \dotsm d_{|w|}\ h(w)$.  Since $w$ and $h(w)$ are nonempty, $|s| \geq 2$.  Set
\[
\begin{matrix*}[l]
	h(c_i) = \lambda & \ \tau(c_i) = q[i] \\
	h(d_i) = \lambda & \ \tau(d_i) = \tau(w[i]) \\
	h(a) = a\ b\ s[3] \dotsm s[|s|] & \ \tau(a) = \tau(s[1]) \\
	h(b) = h(s[1] s[2]) & \ \tau(b) = \tau(s[2])
\end{matrix*}
\]
Then for all $n \geq 1$, we have
\begin{align*}
\tau(h^n(a)) &= \tau(\prod_{i=0}^{n-1} h^i(s) ) \\
&= \tau( s \prod_{i=1}^{n-1} h^i(s) ) \\
&= \tau( s \prod_{i=2}^{n} h^i(w) ) \\
&= q\ \tau( w\ h(w) \prod_{i=2}^{n} h^i(w) ) \\
&= q \prod_{i=0}^n \tau(h^i(w)) \\
&= q 		\prod_{i=1}^{n+1}					R(l,i)
\end{align*}
and therefore $h$ is prolongable on $a$ and $\tau(h^\omega(a)) = \alpha$.  Now, $h$ has growth $\Theta(n^{k-1})$ on $w$, and $w$ is not mortal under $h$.  So $|h(w)|$ is bounded above and below by polynomials of degree $k-1 \geq 0$.  Therefore $h$ has growth $\Theta(n^k)$ on $a$, since for any polynomial $p(n)$ of degree $k$, the sum $\sum_{i=0}^n p(i)$ equals a polynomial $p'(n)$ of degree $k+1$.  Therefore $\alpha$ is a morphic word with growth $\Theta(n^k)$.
\end{proof}

\subsection{From morphic words to zigzag words}

We begin by defining the rank and level of a string under a morphism, as well as the concept of a normalized morphism.

\paragraph{Rank and level}

Let $h$ be a morphism.  Following \cite{er79}, we now define the \textbf{rank} of a letter $c$ under $h$, denoted $\rank(c,h)$.  Informally, the rank 0 letters are those that are finite under $h$, the rank 1 letters are those that are finite under $h$ once the rank 0 letters have been removed, the rank 2 letters are those that are finite under $h$ once the rank 0 and rank 1 letters have been removed, and so on.  Formally, for an alphabet $B$, let $\varphi(h,B)$ be the morphism such that for all $d \in A$, $\varphi(h,B)(d)$ is the string resulting from $h(d)$ by erasing all letters in $h(d)$ that are not in $B$.  We define $\rank(c,h)$ as follows:

\begin{itemize}
	\item If $L(c,h)$ is finite, then $\rank(c,h) = 0$.
	\item For $n \geq 1$, let $A_n$ = $A - \{d \mid \rank(d,h) < n\}$ and let $h_n = \varphi(h,A_n)$.  If $c \in A_n$ and $L(c,h_n)$ is finite, then $\rank(c,h) = n$.
\end{itemize}

Note that some letters may have no rank under $h$.  A string $x$ has rank under $h$ if $x \neq \lambda$ and every letter in $x$ has rank under $h$.  The rank of $x$ under $h$, denoted $\rank(x,h)$, is $\max \{\rank(c,h) \mid$ the letter $c$ occurs in $x\}$.  In \cite{er79} it is shown that $x$ has rank $k$ under $h$ iff $k$ is the minimal degree of a polynomial $p$ such that for every $n$, $p(n) \geq |h^n(x)|$.

We now introduce a more fine-grained concept of rank, called \textbf{level}.  For each letter $c$ with rank under $h$, we define $\level(c,h)$ as follows.  If $c$ is mortal, $\level(c,h) = 0$.  Otherwise, if $c$ is recursive (reachable from itself), $\level(c,h) = \rank(c,h) \cdot 2 + 1$.  Otherwise, $\level(c,h) = \rank(c,h) \cdot 2 + 2$.

Thus, a rank 0 letter may have level 0, 1, or 2, a rank 1 letter may have level 3 or 4, a rank 2 letter may have level 5 or 6, and so on.

A string $x$ has level under $h$ iff it has rank under $h$.  The level of $x$ under $h$, denoted $\level(x,h)$, is $\max \{\level(c,h) \mid$ the letter $c$ occurs in $x\}$.

If the intended morphism is clear from context, we write $\rank(c)$ instead of $\rank(c,h)$, $\level(x)$ instead of $\level(x,h)$, etc.

\paragraph{Normalized morphism}

Below we make use of the concept of a normalized morphism from \cite{dk2009}.  A \textbf{normalized} morphism $h$ has the following properties (among others which we omit):
\begin{itemize}
	\item $\alp(h(c)) = \alp(h^2(c))$ for all $c \in A$
	\item $h(c) = h^2(c)$ for all $c \in A$ such that $\rank(c,h) = 0$
\end{itemize}
By \cite[Lemma 17]{dk2009}, for every morphism $h$, there is a power $h' = h^t$ with $t \geq 1$ such that $h'$ is normalized.

\begin{lemma}\label{lemma-basic}Let $h$ be a morphism and let $c$ be a letter with rank under $h$.  Then $\rank(h(c)) = \rank(c)$.\end{lemma}
		
\begin{proof}
Let $i = \rank(c)$.  For $n \geq 1$, let $A_n = A - \{b \mid \rank(b,h) < n\}$ and let $h_n = \varphi(h,A_n)$.  Suppose $h(c)$ contains a letter $d$ of rank $> i$.  Since $\rank(c) = i$, $L(c,h_i)$ is finite.  Since $\rank(d) > i$, $L(d,h_i)$ is infinite.  But $h_i(c)$ contains $d$, making $L(c,h_i)$ infinite, a contradiction.  So every letter in $h(c)$ has rank $\leq i$.  Then if $i = 0$, we have $\rank(h(c)) = 0$.  So say $i > 0$.  If every letter in $h(c)$ has rank $< i$, then every letter in $h_{i-1}(c)$ has rank $i-1$.  Then for every letter $d$ in $h_{i-1}(c)$, $L(d,h_{i-1})$ is finite.  But then $L(h_{i-1}(c),h_{i-1})$ is finite; hence $L(c,h_{i-1})$ is finite, and therefore $\rank(c) = i-1$, a contradiction.  So at least one letter in $h(c)$ has rank $i$, and therefore $\rank(h(c)) = i$.
\end{proof}

\begin{lemma}\label{lemma-monorec}Let $h$ be a morphism, let $c$ be a letter with rank under $h$, and suppose $\rank(c) \geq 1$ and $h(c) = sct$ for some $s,t \in A^*$.  Then $\rank(st) = \rank(c)-1$.\end{lemma}

\begin{proof}
For $n \geq 1$, let $A_n = A - \{b \mid \rank(b,h) < n\}$ and let $h_n = \varphi(h,A_n)$.  Let $i = \rank(c)$ and $j = \rank(st)$.  Suppose $j < i-1$.  Then $h_{i-1}(c) = c$, so $L(c,h_{i-1})$ is finite.  But then $\rank(c) = i-1$, a contradiction.  So $j \geq i-1$.  Now, since $L(c,h_i)$ is finite, $st$ must be mortal under $h_i$.  So for some $k \geq 0$, $h_i^k(st) = \lambda$.  Then $h_{i-1}^k(st)$ consists entirely of letters of rank $i-1$.  But then $L(h_{i-1}^k(st), h_{i-1})$ is finite; hence $L(st, h_{i-1})$ is finite.  Therefore since $j \geq i-1$, applying the definition of rank gives $j = i-1$.
\end{proof}

\begin{lemma}\label{lemma-letter}Let $h$ be a normalized morphism and let $c$ be a letter that has rank under $h$ and is not mortal under $h$.  Suppose $h(c)$ does not include $c$.  Then $\level(h(c)) = \level(c)-1$.\end{lemma}

\begin{proof}
Since $h$ is normalized, $\alp(h(c)) = \alp(h^2(c))$.  Then for all $i \geq 1$, $\alp(h^i(c)) = \alp(h(c))$.  So since $h(c)$ does not include $c$, there is no $i \geq 1$ such that $h^i(c)$ includes $c$.  Therefore $c$ is not recursive.  Hence since $c$ is not mortal, $\level(c) = \rank(c) \cdot 2 + 2$.  Now, by Lemma \ref{lemma-basic}, $\rank(h(c)) = \rank(c)$, so $h(c)$ contains at least one letter $d$ such that $\rank(d) = \rank(c)$.  Take any such $d$ that is not mortal.  Then since $h$ is normalized, $d$ appears in $h^{|A|+2}(c)$.  So $d$ can be reached from $c$ via a chain of $|A|+1$ ancestor letters, each of $\rank(c)$.  Some letter $e$ appears twice in this chain.  Therefore $e$ is recursive.  Then since $h$ is normalized, $h(e) = set$ for some $s,t \in A^*$.  If $\rank(e) = 0$, then $st$ is mortal, and if $\rank(e) \geq 1$, then by Lemma \ref{lemma-monorec}, $\rank(st) = \rank(c)-1$.  In either case, $d$ cannot appear in $st$ or be descended from $st$, and therefore $d = e$.  Hence $d$ is recursive.  This holds for every letter $d$ with $\rank(c)$ in $h(c)$ that is not mortal.  Therefore $\level(h(c)) = \rank(c) \cdot 2 + 1 = \level(c)-1$.
\end{proof}

\begin{lemma}\label{lemma-string-0}Let $h$ be a normalized morphism and let $x$ be a string with rank 0 and level $> 0$ under $h$.  Take any $v \geq \level(x)$.  Then there is an $l \in L$ of depth $1$ such that for all $i \geq 1$, $R(l,i) = h^{v+i}(x)$.\end{lemma}

\begin{proof}
Since $\rank(x) = 0$, $L(x,h)$ is finite.  Let $s = h(x)$.  Since $\level(x) > 0$, $x$ is not mortal under $h$, so $s \neq \lambda$.  Further, since $h$ is normalized, $h^i(x) = s$ for all $i \geq 1$.  So let $l = [(S,s)]$.  Then $l$ has depth 1 and for all $i \geq 1$, $R(l,i) = s = h^{v+i}(x)$.
\end{proof}

\begin{lemma}\label{lemma-string}Let $h$ be a normalized morphism and let $x$ be a string with rank under $h$ such that $\level(x) > 0$.  Let $r = \rank(x)$ and take any $v \geq \level(x)$.  Then there is an $l \in L$ of depth $r+1$ such that for all $i \geq 1$, $R(l,i) = h^{v+i}(x)$.\end{lemma}


\begin{proof}

We proceed by induction on $v$.  For the base case of $v = 1$, we have $\level(x) = 1$.  Then $r = 0$, so the claim holds by Lemma \ref{lemma-string-0}.

So say $v \geq 2$.  Suppose for induction that for all $v' < v$, for every string $x'$ with rank under $h$ such that $v' \geq \level(x') > 0$, there is an $l' \in L$ of depth $\rank(x')+1$ such that for all $i \geq 1$, $R(l',i) = h^{v'+i}(x')$.

For each $j$ from 1 to $|x|$ such that $x[j]$ is not mortal, we will construct an $l_j \in L$ of depth $\rank(x[j])+1$ such that for all $i \geq 1$, $R(l_j,i) = h^{v+i}(x[i])$.

To construct $l_j$, proceed as follows.  Let $c = x[j]$.  Suppose $h(c)$ does not include $c$.  Since $c$ is not mortal, we have $\level(h(c)) > 0$.  Further, we can apply Lemma \ref{lemma-letter}, obtaining $\level(h(c)) = \level(c)-1$.  Let $v' = v-1$ and $x' = h(c)$.  Then $v' \geq \level(x') > 0$, so we can apply the induction hypothesis, obtaining an $l' \in L$ of depth $\rank(h(c))+1$ such that for all $i \geq 1$, $R(l',i) = h^{v-1+i}(h(c))$.  Then for all $i \geq 1$, $R(l',i) = h^{v+i}(c)$ as desired.  Further, by Lemma \ref{lemma-basic}, $\rank(h(c)) = \rank(c)$, so $\depth(l') = \rank(c)+1$ as desired.  So set $l_j$ to $l'$.



So say $h(c)$ includes $c$.  Then $h(c) = sct$ for some $s,t \in A^*$.  If $\rank(c) = 0$, then the claim holds by Lemma \ref{lemma-string-0}.  So say $\rank(c) \geq 1$.  Then we can apply Lemma \ref{lemma-monorec}, obtaining $\rank(st) = \rank(c)-1$.  Suppose that neither $s$ nor $t$ is mortal.  Then by the induction hypothesis, taking $v' = v-1$ and $x' = s$, there is an $l_s \in L$ of depth $\rank(s)+1$ such that for all $i \geq 1$, $R(l_s,i) = h^{v-1+i}(s)$.  Similarly, taking $x' = t$, there is an $l_t \in L$ of depth $\rank(t)+1$ such that for all $i \geq 1$, $R(l_t,i) = h^{v-1+i}(t)$.  Now set
\begin{align*}
	l_j = [(B,l_s), (S,h^v(c)), (F,l_t)]
\end{align*}

We have for all $n \geq 1$,
\begin{align*}
	R(l_j, n) &= B(l_s,n)\ \ \ \ \ h^v(c)\ \ \ \ \ F(l_t,n) \\
	&= \prod_1^{i=n} R(l_s,i)\ \ \ \ \ h^v(c)\ \ \ \ \ \prod_{i=1}^n R(l_t,i) \\
	&= \prod_1^{i=n} h^{v-1+i}(s)\ \ \ \ \ h^v(c)\ \ \ \ \ \prod_{i=1}^n h^{v-1+i}(t) \\
	&= \prod_v^{i=v+n-1} h^i(s)\ \ \ \ \ h^v(c)\ \ \ \ \ \prod_{i=v}^{v+n-1} h^i(t) \\
	&= \prod_v^{i=v+n-1} h^i(s)\ \ \ \ \ \prod_0^{i=v-1} h^i(s)\ \ \ \ \ c\ \ \ \ \ \prod_{i=0}^{v-1} h^i(t)\ \ \ \ \ \prod_{i=v}^{v+n-1} h^i(t) \\
	&= \prod_0^{i=v+n-1} h^i(s)\ \ \ \ \ c\ \ \ \ \ \prod_{i=0}^{v+n-1} h^i(t) \\
	&= h^{v+n}(c)
\end{align*}
Further, $\depth(l_j) = \max(\depth(l_s),\depth(l_t))+1 = (\rank(st)+1)+1 = \rank(c)+1$, as desired.  Now, we supposed above that neither $s$ nor $t$ was mortal.  They cannot both be mortal, since that would make $\rank(c) = 0$, and we are considering the case $\rank(c) \geq 1$.  So suppose that $s$ is mortal and $t$ is not.  Since $h$ is normalized, $h(s) = \lambda$.  So construct $l_t$ as above and set 
\begin{align*}
l_j = [(S,h^v(c)), (F,l_t)]
\end{align*}
Then following the derivation above, we have for all $n \geq 1$,
\begin{align*}
R(l_j, n) &= h^v(c) \prod_{i=v}^{v+n-1} h^i(t) \\
&= \prod_0^{i=v-1} h^i(s)\ \ \ \ \ c\ \ \ \ \ \prod_{i=0}^{v-1} h^i(t)\ \ \ \ \ \prod_{i=v}^{v+n-1} h^i(t) \\
&= s\ c\prod_{i=0}^{v+n-1} h^i(t) \\
&= h^{v+n}(c)
\end{align*}
and $\depth(l_j) = \depth(l_t)+1 = (\rank(t)+1)+1 = \rank(c)+1$ as desired.  Similarly, if $t$ is mortal and $s$ is not, then $h(t) = \lambda$.  So construct $l_s$ as above and set 
\begin{align*}
l_j = [(B,l_s), (S,h^v(c))]
\end{align*}
Then following the original derivation, we have for all $n \geq 1$,
\begin{align*}
R(l_j, n) &= \prod_v^{i=v+n-1} h^i(s)\ \ \ \ \ h^v(c) \\
&= \prod_v^{i=v+n-1} h^i(s)\ \ \ \ \ \prod_0^{i=v-1} h^i(s)\ \ \ \ \ c\ \ \ \ \ \prod_{i=0}^{v-1} h^i(t) \\
&= \prod_0^{i=v+n-1} h^i(s)\ \ \ \ \ c\ \ \ \ \ t \\
&= h^{v+n}(c)
\end{align*}
and $\depth(l_j) = \depth(l_s)+1 = (\rank(s)+1)+1 = \rank(c)+1$ as desired.

We have now constructed an $l_j$ for each $j$ from 1 to $|x|$ such that $x[j]$ is not mortal.  Since $h$ is normalized, $h(x[j]) = \lambda$ if $x[j]$ is mortal.  So for each $j$ from 1 to $|x|$ such that $x[j]$ is mortal, set $l_j = [\ ]$, the empty list.  Now let $l = \append(l_1, \dotsc, l_{|x|})$.  We have that for all $i \geq 1$, $R(l,i) = h^{v+i}(x)$, as desired.  Further, $l$ has depth
\begin{align*}
&\max \{\depth(l_j) \mid 1 \leq j \leq |x| \text{ and $x_j$ is not mortal} \} \\
&= \max \{\rank(x[j])+1) \mid 1 \leq j \leq |x|\} \\
&= \rank(x)+1
\end{align*} 
as desired.


\end{proof}

\begin{theorem}\label{morphic-zigzag}Every morphic word with growth $\Theta(n^k)$ is a zigzag word of depth $k$.\end{theorem}

\begin{proof}
Take any morphic word $\alpha$ with growth $\Theta(n^k)$.  There exist a morphism $h$, coding $\tau$, and letter $c$ such that $h$ is prolongable on $c$, $\alpha = \tau(h^\omega(c))$, and $h$ has growth $\Theta(n^k)$ on $c$.

By \cite[Lemma 17]{dk2009}, there is a power $h' = h^t$ with $t \geq 1$ such that $h'$ is normalized.  Because $h$ has growth $\Theta(n^k)$ on $c$ and $h'$ is a power of $h$, $h'$ also has growth $\Theta(n^k)$ on $c$.  Since $h$ is prolongable on $c$, $h'$ is prolongable on $c$.  So $h'(c) = cx$ for some $x \in A^*$.  Then
\begin{align*}
	\alpha = \tau(c\ x\ h'(x)\ h'^{\,2}(x)\ \dotsm)
\end{align*}

Let $v = \level(x,h')$.  Let
\begin{align*}
	q = \tau(h'^{\,v+1}(c)) = \tau(	c		\prod_{i=0}^v		h'^{\,i(x)} )
\end{align*}

By \cite[Theorem 3]{er79}, $\rank(c,h') = k$.  Since $\alpha$ is infinite, we know $k \geq 1$.  Hence by Lemma \ref{lemma-monorec}, $\rank(x,h') = k-1$.  Again since $\alpha$ is infinite, $\level(x,h') > 0$.  Then we can apply Lemma \ref{lemma-string}, obtaining an $l \in L$ of depth $k$ such that for all $i \geq 1$, $R(l,i) = h'^{\,v+i}(x)$.  Let $l'$ be $l$ with every string $s$ replaced with $\tau(s)$.  Then for all $i \geq 1$, $R(l',i) = \tau(h'^{\,v+i}(x))$.  Then
\begin{align*}
	q 		\prod_{i \geq 1}			R(l',i) &= \tau(h'^{\,v+1}(c)) 		\prod_{i \geq 1}			\tau(h'^{\,v+i}(x)) \\
	&= \tau(	c	\prod_{i=0}^v		h'^{\,i}(x) )			 \prod_{i \geq v+1}			\tau(h'^{\,i}(x)) \\
	&= \tau(	c 	\prod_{i \geq 0}			h'^{\,i}(x) ) \\
	&= \alpha
\end{align*}
Therefore $\alpha$ is a zigzag word of depth $k$.
\end{proof}

\subsection{Main result}

Finally, we obtain our main result.

\begin{theorem}\label{theorem-main}An infinite word is morphic with growth $\Theta(n^k)$ iff it is a zigzag word of depth $k$.\end{theorem}

\begin{proof}
Immediate from Theorems \ref{zigzag-morphic} and \ref{morphic-zigzag}.
\end{proof}

\section{Applications}\label{sec:applications}

In this section we apply the equivalence between zigzag words and morphic words with polynomial growth obtained in the previous section to the first two orders of growth.  We show an exact correspondence between morphic words with growth $O(n)$ and ultimately periodic words, and between morphic words with growth $O(n^2)$ and multilinear words.  As far as we are aware, these results have not previously appeared in the literature.

\begin{theorem}\label{ultper}An infinite word is morphic with growth $O(n)$ iff it is ultimately periodic.\end{theorem}

\begin{proof}
\noindent $\Longrightarrow$: Take any morphic word $\alpha$ with growth $O(n)$.  By Proposition \ref{prop1}, $\alpha$ has growth $\Theta(n)$.  (Since $\alpha$ is infinite, it cannot have growth $\Theta(1)$.)  Then by Theorem \ref{theorem-main}, $\alpha$ is a zigzag word of depth 1.  Then there are $q \in A^*$ and $l \in L$ such that
\begin{align*}
	\alpha = 	q 		\prod_{i \geq 1}			R(l,i)
\end{align*}
and $l$ has depth 1.  Then $l$ has the form $[(S,r_1), \dotsc, (S,r_m)]$.  It follows that $\alpha = q (r_1 \dotsm r_m)^\omega$.  Thus $\alpha$ is ultimately periodic.

\bigskip

\noindent $\Longleftarrow$: Take any ultimately periodic word $\alpha = qr^\omega$.  Let $l = [(S,r)]$.  Then 
\begin{align*}
	\alpha = 	q 		\prod_{i \geq 1}			R(l,i)
\end{align*}
So $\alpha$ is a zigzag word of depth 1.  Then by Theorem \ref{theorem-main}, $\alpha$ is a morphic word with growth $O(n)$.
\end{proof}

\begin{theorem}An infinite word is morphic with growth $O(n^2)$ iff it is multilinear.\end{theorem}

\begin{proof}
\noindent $\Longrightarrow$: Take any morphic word $\alpha$ with growth $O(n^2)$.  By Proposition \ref{prop1}, $\alpha$ has growth $\Theta(n)$ or $\Theta(n^2)$.  If $\alpha$ has growth $\Theta(n)$, then by Theorem \ref{ultper}, it is ultimately periodic and hence multilinear.  So say $\alpha$ has growth $\Theta(n^2)$.  Then by Theorem \ref{theorem-main}, $\alpha$ is a zigzag word of depth 2.  Then there are $q \in A^*$ and $l \in L$ such that
\begin{align*}
\alpha = 	q 		\prod_{i \geq 1}			R(l,i)
\end{align*}
and $l$ has depth 2.  Now, $l$ has the form $[(f_1,x_1), \dotsc, (f_m,x_m)]$.  For each $1 \leq i \leq m$, we create a term $t_i = [r_i,a_i,b_i]$ as follows.  If $f_i = S$, then $x_i$ is a string, so let $t_i = [x_i, 0, 1]$.  If $f_i = F$ or $B$, then $x_i$ has depth 1, and therefore has the form $[(S,s_1),\dotsc,(S,s_k)]$.  So let $t_i = [s_1 \dotsm s_k, 1, 0]$.  Now the multilinear word $[q, [t_1, \dotsc, t_m]]$ equals
\begin{align*}
	q 	\prod_{n \geq 1}		\prod_{i=1}^m			r_i^{n \cdot a_i + b_i} &= q 	\prod_{n \geq 1}	\prod_{i=1}^m			f_i(x_i,n) \\
	&= q 	\prod_{n \geq 1}		R(l,i) \\
	&= \alpha
\end{align*}
as desired.

\bigskip

\noindent $\Longleftarrow$: Take any multilinear word $\alpha$.  If $\alpha$ is ultimately periodic, then by Theorem \ref{ultper}, $\alpha$ is a morphic word with growth $O(n)$.  Otherwise, $\alpha$ is properly multilinear, so by \cite[Theorem 7]{smith2016b}, we can write $\alpha$ as
\begin{align*}
	q 	\prod_{n \geq 1}		\prod_{i=1}^m			p_i s_i^n
\end{align*}
for some $m \geq 1$, $q \in A^*$, and $p_i, s_i \in A^+$.  Let
\begin{align*}
l = [(S,p_1),(F,[(S,s_1)]),\ \dotsc\ ,(S,p_m),(F,[(S,s_m)])].
\end{align*}
Then $l$ has depth 2 and
\begin{align*}
	\alpha	= q 	\prod_{n \geq 1}			R(l,i)
\end{align*}
So $\alpha$ is a zigzag word of depth 2.  Then by Theorem \ref{theorem-main}, $\alpha$ is a morphic word with growth $O(n^2)$.
\end{proof}

\section{Conclusion}\label{sec:conclusion}

In this paper we characterized morphic words with polynomial growth in terms of zigzag words, showing that an infinite word is morphic with growth $\Theta(n^k)$ iff it is a zigzag word of depth $k$.  We then applied this characterization to show that the morphic words with growth $O(n)$ are exactly the ultimately periodic words, and the morphic words with growth $O(n^2)$ are exactly the multilinear words.

Some open problems involving the above characterization arise in connection with automata.  We say that an automaton $M$ determines an infinite word $\alpha$ if $L(M)$ is infinite and every string in $L(M)$ is a prefix of $\alpha$.  In \cite{smith3} it is shown that ultimately periodic words are exactly those determined by finite automata, and multilinear words are exactly those determined by one-way stack automata (a generalization of pushdown automata).  It would be interesting to know what kind of automaton determines exactly the zigzag words.

It is further shown in \cite{smith3} that every multilinear word can be determined by a one-way 2-head DFA.  (These automata can also determine infinite words that are not multilinear.)  It would be interesting to know whether or not the following statement holds: every zigzag word of depth $k$ (morphic word with $\Theta(n^k)$ growth) can be determined by a one-way $k$-head DFA.

\acknowledgements

I would like to thank Jeffrey Shallit for suggesting the connection between multilinear words and morphic words with quadratic growth, and for his helpful comments on a draft of this paper.

\bibliographystyle{alpha}
\bibliography{../sources}

\end{document}